\documentclass[12pt]{article}
\usepackage{amsmath}
\usepackage{csvsimple}
\usepackage{geometry}
\usepackage{subfig} 
\usepackage{amsthm}
\usepackage{amssymb}
\usepackage{mathtools}
\newtheorem{theorem}{Theorem}
\newtheorem{lemma}{Lemma}
\newtheorem{corollary}{Corollary}

\DeclareMathOperator*{\argmin}{arg\,min}
\newcommand\norm[1]{\left\lVert#1\right\rVert}
\def\restrict#1{\raise-.5ex\hbox{\ensuremath|}_{#1}}
\usepackage{placeins} 
\usepackage{dsfont} 
\usepackage{xcolor}

\usepackage{graphicx}

\newcommand{\blind}{0}

\usepackage[
backend=biber,
style=authoryear
]{biblatex}

\AtEveryBibitem{
    \clearfield{urlyear}
    \clearfield{urlmonth}
}
\begin{document}

\def\spacingset#1{\renewcommand{\baselinestretch}%
{#1}\small\normalsize} \spacingset{1}


\if0\blind
{
  \title{\bf Sparse-group boosting: \\
Unbiased group and variable selection}
  \author{Fabian Obster\thanks{
    Funding was provided by dtec.bw funded by NextGenerationEU. }\hspace{.2cm}\\
    Department of Business Administration,\\ University of the Bundeswehr Munich, Germany\\
    and \\
    Christian Heumann \\
    Department of Statistics,\\ Ludwig Maximilians University Munich, Germany}
  \maketitle
} \fi

\if1\blind
{
  \bigskip
  \bigskip
  \bigskip
  \begin{center}
    {\LARGE\bf Title}
\end{center}
  \medskip
} \fi

\bigskip
\begin{abstract}
For grouped covariates, we propose a framework for boosting that allows for sparsity within and between groups. By using component-wise and group-wise gradient ridge boosting simultaneously with adjusted degrees of freedom or penalty parameters, a model with similar properties as the sparse group lasso can be fitted through boosting. We show that within-group and between-group sparsity can be controlled by a mixing parameter, and discuss similarities and differences to the mixing parameter in the sparse group lasso. Furthermore, we show under which conditions variable selection on a group or individual variable basis happens and provide selection bounds for the regularization parameters depending solely on the singular values of the design matrix in a boosting iteration of linear Ridge penalized boosting. In special cases, we characterize the selection chance of an individual variable vs. a group of variables through a generalized beta prime distribution. With simulations as well as two real datasets from ecological and organizational research data, we show the effectiveness and predictive competitiveness of this novel estimator. The results suggest that in the presence of grouped variables, sparse group boosting is associated with less biased variable selection and higher predictability compared to component-wise or group-component-wise boosting.
\end{abstract}

\noindent%
\vfill

\newpage
\spacingset{1.5} 
\section{Introduction}
A key task in empirical science in the presence of high-dimensional data is to perform variable selection. Especially if the number of variables is relatively high compared to the number of observations. In biostatistics, this is a common setting, for example, in gene sequencing (\cite{johnstone_statistical_2009}). Two common variable selection strategies are the use of a lasso penalty (\cite{tibshirani_regression_1996}) or component-wise boosting (\cite{breiman_arcing_1998}; \cite{friedman_additive_2000}). Many strategies exist to find a subset of data, including forward selection, backward elimination, or even all-possible subset selection (\cite{chowdhury_variable_2020}), where all possible combinations of variables are considered. Methods differ not only by the selection strategy but also by the metric determining the resulting subset of variables. Some include second-generation p-values (\cite{zuo_variable_2022}) others modified loss functions leading to sparsity through shrinkage like the lasso. Often, the variables in the data can be clustered into groups. These could be pathways of genes or items of a construct in a questionnaire, used, for example, in the social sciences or psychology (\cite{gogol_my_2014}; \cite{agarwal_verifying_2011}). In these cases, it can be of interest to perform variable selection in such a way that this group structure is accounted for. Through the group lasso penalty (\cite{yuan_model_2006}; \cite{meier_group_2008}) and group-wise boosting (\cite{kneib_variable_2009}) this can be achieved. A solution where variable selection is based on groups, as well as variables, can be of interest if one wants to identify important groups as well as important variables within a group or additionally to a group. This can be achieved by the sparse group lasso (\cite{simon_sparse-group_2013}). Most applications of datasets with sparse group structures rely on the utilization of the sparse group lasso penalty in some form, like sparse group quantile regression (\cite{mendez-civieta_adaptive_2021}), sparse group neural networks (\cite{yoon_combined_2017}) and support vector machines (\cite{tang_group_2018}). One exception is sparse group Bayesian regression (\cite{chen_bayesian_2016}). However, to our knowledge, an in-depth analysis of such sparse group variable selection in the context of boosting has not been conducted. Since boosting is a widely utilized machine learning algorithm, a boosting variation that can deal with sparse group structures can offer an alternative modeling approach beyond the sparse group lasso. Having an alternative to the sparse group lasso is especially important since many Machine Learning systems use (sparse group variable) selection methods prior to (\cite{farokhmanesh_deep_2019}) or after (\cite{zhao_heterogeneous_2015}) fitting another machine learning algorithm, leaving the sparse group variable selection algorithm a potential bottleneck for predictive power and interpretability. In this manuscript, we show the issues and potential biases, as well as their correction, occurring in the presence of variable selection between and within groups in the context of boosting. We provide an algorithm for sparse group boosting and discuss its advantages over alternative definitions. Differences and similarities between the sparse group lasso and the sparse group boosting are described with special attention to sparsity. We apply sparse group boosting to an organizational dataset as well as to an agricultural dataset and compare its results to component-wise, group component-wise boosting, and sparse group lasso. The same comparison will be conducted with extensive simulations. The code used for the analysis and figure creation as well as the raw data is available at GitHub (https://github.com/FabianObster/sgb).
%
%
\subsection{Notation and general setup}
Throughout this article, we consider a (generalized) linear regression framework with outcome $y \in \mathbb{R}^n$, design matrix $X$ consisting of $n$ observations and $p$ variables. The $p$ variables are grouped in $G$ non-overlapping groups, where each group $g \in \{1,...,G\}$ consists of $p_g$ variables. We refer to the $j$-th variable in the $g$-th group as $x^{(g)}_j$ and the sub-matrix containing only the columns belonging to group $g$ is denoted as $X^{(g)}$. If groups are not considered, the group index $g$ is omitted. The same notation also applies to the parameter vector $\beta \in \mathbb{R}^p$ and regularization parameters.

%
%

\subsection{The sparse group Lasso}
In a possibly high-dimensional setting, e.g. $p \gg n$, the sparse group lasso can fit a model that not only performs variable selection on a variable basis but also on a group basis (\cite{simon_sparse-group_2013}).
The sparse group lasso achieves this by combining the group lasso penalty $ \sum_{g=1}^G \sqrt{p_g} \norm{ \beta^{(g)}}_2$ (\cite{yuan_model_2006})
and the lasso penalty $\norm{\beta}_1$ (\cite{tibshirani_regression_1996}) with a mixing parameter $\alpha \in [0,1]$,
\begin{equation}\label{sgl_loss}
    \text{min}_{\beta} \frac{1}{2n} \norm{y- \sum_{g=1}^G X^{(g)}\beta^{(g)}}_2^2 + (1-\alpha)\lambda \sum_{g=1}^G \sqrt{p_g} \norm{ \beta^{(g)}}_2 + \alpha \lambda \norm{\beta}_1.
\end{equation}
There are two tuning parameters: $\alpha$ and $\lambda \geq 0$. The mixing parameter $\alpha$ decides how much we want to penalize the individual variables (increase $\alpha$) versus how much we want to penalize groups (decrease $\alpha$). The special case of $\alpha = 1$ yields the lasso fit, and $\alpha = 0$ yields the group-lasso fit. As in (\cite{simon_sparse-group_2013}), we differentiate between the two types of sparsity: "within-group sparsity" refers to the number of non-zero coefficients within each non-zero group, and "group-wise sparsity" refers to the number of groups with at least one non-zero coefficient. 
Depending on $\alpha$, both types of sparsity can be balanced. This gives the data scientist the flexibility to include secondary knowledge regarding the two types of sparsity. If $\alpha$ is not known before, it has to be estimated, for example, by using cross-validation on a two-dimensional grid for $\lambda$ and $\alpha$. This has the downside that two hyperparameters have to be tuned. \\
The sparse group lasso can also be extended to more general loss functions by replacing the least squared loss with other loss functions. This way, also generalized linear models can be fitted by using the negative log-likelihood $l(\beta)$, with group-wise and within-group sparsity
\begin{equation*}
    \text{min}_{\beta} l(\beta) + (1-\alpha)\lambda \sum_{g=1}^G \sqrt{p_g} \norm{ \beta^{(g)}}_2 + \alpha \lambda \norm{\beta}_1.
\end{equation*}

%
%
\subsection{Model-based boosting}

Another way of fitting sparse regression models is through the method of boosting. The fitting strategy is to continuously improve a given model by adding a base-learner to it. Throughout this article, we refer to a base-learner as a subset of columns of the design matrix associated with a real-valued function.
To enforce sparsity, each base-learner only considers a subset of the variables at each step (\cite{buhlmann_boosting_2007}). In the case of component-wise $\mathcal{L}^2$ boosting, each variable will be a base-learner with a linear link function. In the case of a one-dimensional B-Spline, a base-learner is the design matrix representing the basis functions of the B-Spline with a linear link function.
The goal of boosting in general is to find a real-valued function that minimizes a typically differentiable and convex loss function $l(\cdot,\cdot)$. Here, we will consider the negative log-likelihood as a loss function to estimate $f^*$ as
\begin{equation*}
f^*(\cdot)=\argmin_{f(\cdot)} \mathbb{E}[l(y,f)].
\end{equation*}

\textit{General functional gradient descent Algorithm} (\cite{friedman_greedy_2001})
\begin{enumerate}
\item Define base-learners of the structure $h: \mathbb{R}^{n \times p} \to \mathbb{R}$
\item Initialize $m=0$ and $\widehat{f}^{(0)} \equiv 0$ or $\widehat{f}^{(0)} \equiv \overline{y}$
\item Set $m = m+1$ and compute the negative gradient $\frac{\partial}{\partial f} l(y,f)$ and evaluate it at $\widehat{f}^{[m-1]}$. Doing this yields the pseudo-residuals $u_1,...,u_n$ with
\begin{equation*} 
u_i^{[m]} = \frac{\partial}{\partial f} l(y_i,f)|_{f = \widehat{f}^{[m-1]}},
\end{equation*}
for all $i = 1,..., n$
\item Fit the base-learner  $h$ with the response $(u_1^{[m]},...,u_n^{[m]})$ to the data. This yields $\widehat{h}^{[m]}$, which is an approximation of the negative gradient
\item Update
\begin{equation*}
\widehat{f}^{[m]} = \widehat{f}^{[m-1]} + \eta \cdot \widehat{h}^{[m]}
\end{equation*}
here $\eta$ can be seen as learning rate with $\eta \in ]0,1[$
\item Repeat steps 3,4 and 5 until $m=M$
\end{enumerate}

\textit{Component-wise and group component-wise boosting} \\
In step 4 of the general functional gradient descent algorithm, the function $h$ is applied to the whole data. In the case of component-wise boosting the base procedure is fitted to each variable in the dataset individually. The update in step 5 is then performed only with the base-learner that yields the lowest negative log-likelihood. 

Through early-stopping, or setting $M$ relatively smaller compared to the number of variables in the dataset, and considering the learning rate $\nu$, a sparse overall model can be fitted.

%
%
\section{Boosting Ridge Regression and preliminary results}
The sparse group boosting as we define it here is based on $\mathcal{L}^2$ regularized regression. Therefore, we first discuss some results for linear Ridge Regression which minimizes $(y-X\beta)^T(y-X\beta)$ with the constraint $\norm{\beta}^2 \leq c$, for a positive constant $c$. Using the Lagrangian form this has an explicit solution $\widehat{\beta}_\lambda = (X^TX+\lambda I)^{-1}X^T y$. We will now discuss results for boosting ridge regression regarding the residual sum of squares (RSS) and the degrees of freedom that will be relevant for the sparse group boosting. Lemma \ref{lemma:RSSridge} allows us to characterize the hat matrix in Ridge Regression using the singular values. The ridge hat matrix will be important to understand the RSS and degrees of freedom, which we need to later define the sparse group boosting and then understand the variable selection mechanism.
\begin{lemma}[Hat matrix in $\mathcal{L}^2$ Ridge Boosting]\label{lemma:RSSridge}
    Consider a design matrix $X \in \mathbb{R}^{n \times p}$ of rank $r \leq p$ with singular value decomposition $X = UDV^{T}$, where $U \in \mathbb{R}^{n\times p}, V \in \mathbb{R}^{p\times p}$ are unitary matrices and $D = \text{diag}(d_1,...,d_r,0,...,0)$ is a diagonal matrix containing the singular values. Let $y \in \mathbb{R}^n$ be the outcome variable and $\widehat{\beta}_{\lambda} = (X^{T}X+\lambda I)^{-1}X^{T}y$ be the Ridge estimate for $\lambda \geq 0$. Then the hat matrix $H_\lambda(m)$ after $m$ boosting steps using a learning rate of $\eta = 1$ is given by
    \begin{equation*}
        H_\lambda(m) = I_n-{(I_n-U\Tilde{D}U^{T})}^{m+1} = \sum_{j=1}^r (1-(1-\tilde{d_j})^{m+1})u_ju_j^{T},
    \end{equation*}
    with $\tilde{D} = \text{diag}(\tilde{d_1},...,\tilde{d_r},0,...,0) =\text{diag}(\frac{d_1^2}{d_1^2+\lambda},...,\frac{d_r^2}{d_r^2+\lambda},0,...,0)$.
\end{lemma}
A derivation can be found in \cite{tutz_boosting_2007}. Note that the RSS does not depend on the orthogonal matrix $V$. Considering the case of only one boosting step, the hat matrix becomes 
\begin{align*}
    H_\lambda \coloneqq H_\lambda(0) =U\Tilde{D}U^{T} = \sum_{j=1}^r \frac{d_j^2}{d_j^2+\lambda}u_ju_j^{T}.
\end{align*}
For the residual sum of squares, this means
\begin{align}
    RSS(\widehat{\beta}_\lambda) &= (y-X\widehat{\beta}_\lambda)^T(y-X\widehat{\beta}_\lambda) = y^T(I-H_\lambda)^2y \nonumber\\ &= y^Ty -y^T(2H_\lambda- H_\lambda^2)y \nonumber\\ &= y^Ty-y^T(2U\Tilde{D}U^{T}-U\Tilde{D}^2U^{T}) y\nonumber\\ &= y^Ty -y^T\Big(\sum_{j=1}^r \Big[2\frac{d_j^2}{d_j^2+\lambda}-\frac{d_j^4}{(d_j^2+\lambda)^2}\Big]u_ju_j^{T}\Big)  y.\label{equation:svd}
\end{align}
Now, we can introduce the degrees of freedom $\text{df}(\lambda)$, which are either defined as the trace of the hat matrix $\tilde{\text{df}}(\widehat{\beta}_\lambda) = \text{tr}(H_\lambda)$ or as $\text{df}(\lambda) = \text{tr}(2H_\lambda-H_\lambda^TH_\lambda)$. As discussed by \cite{hofner_framework_2011} and apparent in (\ref{equation:svd}), $\text{df}$ has the advantage over $\tilde{\text{df}}$ of being tailored to the $RSS$. It is worth pointing out that regularizing based on $\text{df}$ leads to a greater shrinkage compared to $\tilde{\text{df}}$, because for the same base-learner 
\begin{equation*}
    \text{df}(\lambda) = \tilde{\text{df}} (\tilde{\lambda}) \Leftrightarrow
    \sum_{j=1}^r 2\frac{d_j^2}{d_j^2+\lambda}-\frac{d_j^4}{(d_j^2+\lambda)^2} = \sum_{j=1}^r\frac{d_j^2}{d_j^2+\tilde{\lambda}}
    \Rightarrow \lambda \geq \tilde{\lambda}.
\end{equation*} 
This can be seen by observing that $2x-x^2 \geq x$ for $x \in [0,1]$.
They conclude that the degrees of freedom should be set equal for all base-learner to avoid selection bias, because 
\begin{equation*}
    \mathbb{E}[RSS(\widehat{\beta}_\lambda)-RSS(\widehat{\gamma}_\mu)] = 0 \Leftrightarrow \text{df}(\lambda) = \text{df}(\mu)
\end{equation*}
if the normally distributed random variables $y \sim \mathcal{N}(0,\sigma I)$ is independent of the design matrices $X_\gamma$ and $X_\beta$ in an ordinary linear regression model. 
However, we do want to point out that having the same expectation of RSS does not mean, that there is no variable selection bias. The RSS of $\widehat{\beta}_\lambda$ can still have a different variance, shape or different higher order moments than the RSS of $\widehat{\gamma}_\mu$. In the same setting, the RSS is a quadratic form and can be written as  $y^TQ_1y$ and $y^TQ_2y$ with symmetric and positive definite matrices $Q_1$ and $Q_2$ for two base-learners. Such quadratic forms are generally not independent of each other unless $Q_1Q_2 = 0$ (Craig's theorem). We will later return to the issue of selection bias in the context of sparse group boosting. We will now look at component-wise Ridge Boosting
\begin{corollary}
     Consider a design matrix vector $x \in \mathbb{R}^{n \times 1}$ of rank one with singular value decomposition $x = ud$, where $u = \frac{x}{\sqrt{x^Tx}}$ is the left singular vector and $d = \sqrt{x^Tx} \in R^+$ is the singular value. Let $y \in \mathbb{R}^n$ be the outcome variable and $\widehat{\beta}_{\lambda} = (x^{T}x+\lambda I)^{-1}x^{T}y$ be the Ridge estimate for $\lambda \geq 0$. Then,
     \begin{align*}
     \text{df}(\lambda) =\text{df}(\widehat{\beta}_{\lambda}) &= \Big(2\frac{d^2}{d^2+\lambda}-\frac{d^4}{(d^2+\lambda)^2}\Big), \\
     RSS(\widehat{\beta}_\lambda) &= y^Ty-y^T\text{df}(\widehat{\beta}_{\lambda})uu^Ty,
         \\
     \end{align*}
     and the Ridge parameter $\lambda$ in terms of $\text{df}(\widehat{\beta}_\lambda)$ is given by
     \begin{align*}
     \lambda = \frac{-\Big(\sqrt{-d^4\big(\text{df}(\lambda)-1\big)}+d^2\big(\text{df}(\lambda)-1\big)\Big)}{\text{df}(\lambda)}.
     \end{align*}
\end{corollary}
This follows directly from Lemma \ref{lemma:RSSridge} and then solving for $\lambda$ by finding the zeros and the fact that $\text{df}(\lambda)$ is greater than zero.
This corollary seems straightforward but has some useful implications. Generally, in model-based boosting, grid search over $\lambda$ has to be performed to set the degrees of freedom to a fixed value. For individual base-learners one can now compute $\lambda$ directly with a simple formula without having to try a lot of regularization parameters, which increases speed and accuracy. In addition, one does not have to compute the singular value of the Demmler-Reinsch orthogonalization (App. B.1.\cite{carroll_computational_2003}), because the singular values of the design matrix are sufficient in this case. We also see that controlling the variance of a covariate can achieve the same effect as regularization in component-wise boosting. Hence, equalizing the degrees of freedom can be seen as a form of standardization. We will now turn to ridge regression with orthogonal design matrices. In this case, the Ridge estimate is equal to a scaled ordinary least squares estimate $\widehat{\beta}_\lambda = \frac{1}{1+\lambda}\widehat{\beta}_{OLS}$. Orthogonal designs also allow us to characterize the difference between the RSS of Ridge regression and the RSS of the OLS estimate as a Gamma distribution.
\begin{theorem}[Distribution of the difference of RSS in orthogonal Ridge regression]\label{theo:ortho_gamma}
Let $X \in \mathbb{R}^{n \times p}$ be a design matrix with orthonormal columns such that $X^TX = I_p$.
Let $y \in \mathbb{R}^n$ be the outcome variable, $y = \epsilon, \epsilon \sim \mathcal{N}(0,\sigma^2)$ not dependent on the design matrix. Further, assume that the least squares estimate $\widehat{\beta} = X^Ty$ exists and $\widehat{\beta}_\lambda$ is the Ridge estimate for some $\lambda>0$.
Define the difference of residual sums of squares as $\Delta =  \text{RSS}(\widehat{\beta}_\lambda) - \text{RSS}(\widehat{\beta}) = (y- \frac{1}{1+\lambda}X \widehat{\beta})^{T}(y-\frac{1}{1+\lambda}X \widehat{\beta}) - (y-X \widehat{\beta})^{T}(y-X \widehat{\beta})$. Then if $\big(1-\frac{\text{df}(\lambda)}{p}\big)> 0$, $\frac{\Delta}{\sigma^2}$ follows a gamma distribution with the following shape-scale parametrization
\begin{equation*}
    \frac{\Delta}{\sigma^2} \sim \Gamma \Big(\frac{p}{2}, 2(1-\frac{2}{p(1+\lambda)}+\frac{1}{p(1+\lambda)^2})\Big).
\end{equation*}
\end{theorem}

%
%
\section{Sparse group boosting}
The goal of this paper is to adapt the concept of the sparse group lasso to the boosting such that it is tailored to the boosting framework. One straightforward idea is to use the whole dataset as base-learner equipped with the sparse group lasso penalty in (\ref{sgl_loss}) and update the global model with each boosting step. However, it is not within the scope of this paper to fit a sparse group lasso model through the utilization of boosting. We rather want to build upon the results from boosting Ridge regression within the framework of group-component-wise boosting. With this approach no Lasso penalty is needed.
As proposed by \cite{hofner_gamboostlss_2014}, one can define one base-learner as a group of variables, as well as an individual variable. Using a similar idea as in the sparse group lasso, we define sparse group boosting.
\subsection{Definition and properties of the sparse group boosting}
Each variable will get its own base-learner, which we will refer to as an individual base-learner, and each group will be one base-learner, containing all variables of the group. For the degrees of freedom of an individual base-learner
$x_j^{(g)} \in \mathbb{R}^{n \times 1}$  we will use
\begin{equation}
     \text{df}(\lambda) =  \alpha,\label{def:sgb_alpha}
\end{equation}
and for the group base-learner we will use
\begin{equation}
     \text{df}(\lambda^{(g)}) = (1-\alpha).
\end{equation}
$\alpha \in ]0,1[$ is the mixing parameter. Since $ \text{df}(\lambda) = 0$ means $\lambda \to \infty$, $\alpha = 1$ yields component-wise boosting, and $\alpha = 0$ yields group boosting.
An alternative definition looking more like the sparse group lasso would be to directly regularize the penalty term instead of the degrees of freedom. In this case the modified loss function L of the unpenalized loss function of the individual base-learner $L^{(g)}_j$ and the one of the group base-learner $L^{(g)}$ become
\begin{align}
   \text{L}^{(g)} = \textit{L}(\beta^{(g)}) + (1-\alpha)\lambda \sqrt{p_g}\norm{\beta^{(g)}}_2 \label{def:sgb_lambda} \\
   \text{L}^{(g)}_j =  \textit{L}(\beta^{(g)}_j) + \alpha \lambda \norm{\beta^{(g)}_j}_2. \\
\end{align}
However, using this definition does not yield either group boosting or component-wise boosting for $\alpha \in \{0,1\}$. This is the case because if we compare the loss function of the regularized base-learner with an unregularized base-learner, it is not guaranteed that the unregularized base-learner has a lower loss. We will see this later in Theorem \ref{th:consistency}.
Still, both definitions have their advantages and disadvantages. Using the degrees of freedom allows us to directly control the expectation of the RSS given normal error terms. In this case $\alpha$ has a natural interpretation and one can set $\alpha$ a priori based on how one wants the RSS of individual base-learners to be compared to the group base-learner. The other advantage of using the degrees of freedom is that one only has to decide on one hyper-parameter, namely $\alpha$. Based on that choice, all other penalty parameters are already determined. Of course, the optimal stopping parameter and the learning rate have to be set in both definitions. There are also advantages of mixing the penalty term. While more tuning is required, there is a greater flexibility of being able to control two parameters independently of each other which may lead to greater predictive power. Controlling the penalty term directly also has the advantage of seeing which combination of $\alpha$ and $\lambda$ leads to either only group or individual variables based on the smallest and biggest non-zero singular value of the design matrix and makes the search more efficient, see Theorem \ref{th:consistency}. In this article, we will mainly focus on the first definition, because of its simplicity, interpretation, and the fact that in boosting the regularization is mainly achieved through the small learning rate and early stopping than finding the optimal regularization parameter as in the sparse group lasso.
\begin{figure}[!ht]
        \centering
        \includegraphics[scale=0.7]{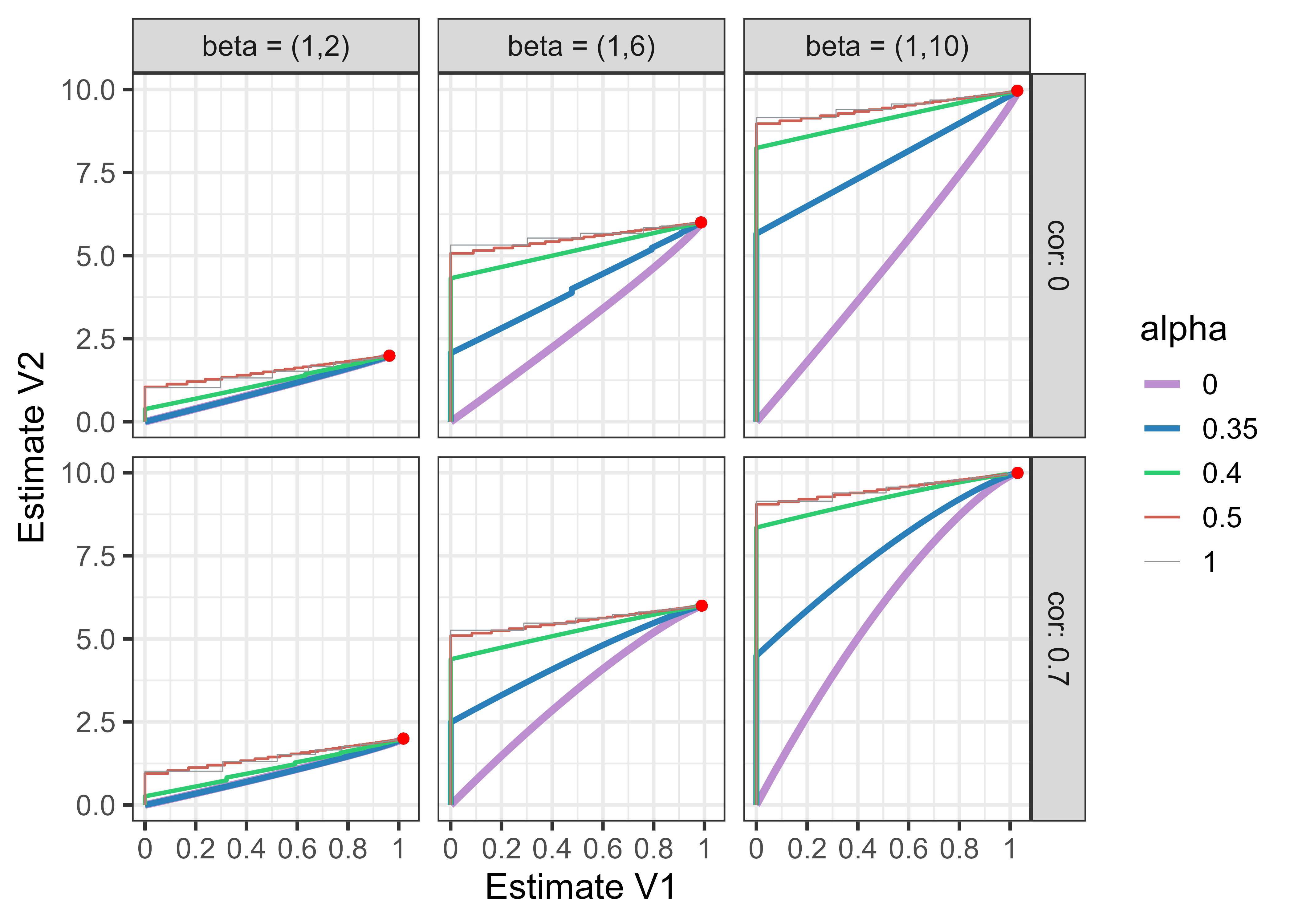}
        \caption{Example: Sparse group boosting parameter estimate paths. Paths throughout 100 boosting iterations for a learning rate of 0.3, depending on the mixing parameter (line-thickness) in an ordinary linear regression model with a normally distributed error term. The point at $\alpha = 1$ indicates the least squares estimate. In the case of $\alpha =0$ group boosting and $\alpha =1$ component-wise boosting was used. }\label{figure:betapaths}
\end{figure}
\FloatBarrier
 Figure \ref{figure:betapaths} displays a two-variable group example of the evolution of the estimates throughout the sparse group boosting process for different mixing parameters based on the degrees of freedom. One group base-learner and two individual base-learners were used. All models move towards the least squares estimate indicated with the point at $\alpha =1$. However, the path they take depends on the mixing parameter. So, in the case of early stopping, different parameter estimates are obtained depending on $\alpha$. One can see that sometimes only the group base-learners are selected, and in some cases, only the individual base-learners are selected, as the path moves only either up or to the right. In some cases, there is an alteration between individual base-learners and the group base-learner. In this example, the Multi-collinearity of the two predictors seems to only slightly affect the selection process, as the upper and lower paths look similar.
 For the sparse group boosting to be a useful method for a general design matrix compared to component-wise boosting and group boosting as separate methods, we show that it is more flexible than just using either of the two.
\begin{theorem}[Selection intervals of the sparse group boosting]\label{th:consistency}
    Consider a design matrix $X \in \mathbb{R}^{n \times p}$ of rank $r \leq p$ with singular value decomposition $X = UDV^{T}$, where $U \in \mathbb{R}^{n\times p}, V \in \mathbb{R}^{p\times p}$ are unitary matrices and $D = \text{diag}(d_1,...,d_r,0,...,0)$ is a diagonal Matrix containing the singular values. Let $y \in \mathbb{R}^n$ be the outcome variable and $\widehat{\beta}_{\mu} = (X^{T}X+\mu I)^{-1}X^{T}y $ be the Ridge estimate for $\mu \geq 0$. For $j \leq p$ let $\widehat{\beta}_{\lambda_j} = (x_j^{T}x_j+\lambda_j)^{-1}x_j^{T}y$ be the estimate for the $j$-th individual base-learner, and $\overline{d}_j$ be the singular value of $x_j$. Denote $\overline{d}^- = \min_{j \leq p} \overline{d}_j^2$ and $\overline{d}^+ = \max_{j \leq p} \overline{d}_j^2$ as well as $d^+ = \max_{j \leq r} d_j^2$ and $d^- = \min_{j \leq r} d_j^2$ accordingly. Then, there are always two mixing parameters $\alpha_1, \alpha_2 \in ]0,1[$ such that
    \begin{align*}
        (\forall_{j \leq p}:  \alpha_1 &= \text{df}(\lambda_j) \wedge (1-\alpha_1) = \text{df}(\mu)) \Rightarrow \min_{j \leq p}{RSS(\lambda_j)} \leq RSS(\mu), \text{and} \\
        (\forall_{j \leq p}: \alpha_2 &= \text{df}(\lambda_j) \wedge (1-\alpha_2) = \text{df}(\mu)) \Rightarrow  RSS(\mu) \leq \min_{j \leq p}{RSS(\lambda_j)}.
    \end{align*}
    Furthermore, the following conditions assure the selection of an individual variable or the whole design matrix
    \begin{align}
        \bigg( \bigg[\forall_{l \leq k} \frac{(d^-+2\mu)}{(d^-+\mu)^2} \leq \frac{(\overline{d}_l^2+2\lambda_l)}{r(\overline{d}_l^2+\lambda_l)^2}\bigg] \lor \bigg[\text{df}(\mu) \leq \frac{\text{df}(\lambda_l)d^-}{r\overline{d}^+}\bigg] \bigg) &\Rightarrow \min_{j \leq p}{RSS(\lambda_j)} \leq RSS(\mu), \label{claim:indi} \\
         \bigg(\bigg[\forall_{l \leq k} \frac{(d^++2 \mu)}{(d^++\mu)^2} \geq \frac{\text{df}(\lambda_l)}{\overline{d}_l^2}\bigg] \lor \bigg[\forall_{l \leq k} \frac{(d^++2 \mu)}{(d^++\mu)^2} \geq \frac{(\overline{d}_l^2+2 \lambda_l)}{(\overline{d}_l^2+\lambda_l)^2}\bigg]) &\Rightarrow  RSS(\mu) \leq \min_{j \leq p}{RSS(\lambda_j)}.\label{claim:group}
    \end{align}
\end{theorem}
Theorem \ref{th:consistency} is useful for both definitions of the sparse group boosting, as one can either use the bounds by setting $\lambda_j = \alpha \lambda, \mu = (1-\alpha) \lambda$ and further bound (\ref{claim:indi}) and (\ref{claim:group}) by replacing $d_l^2$ with either $\overline{d}^+$ or $\overline{d}^-$ respectively. Note that in (\ref{claim:group}) it is not possible to use the degrees of freedom of the group design matrix as a bound as in (\ref{claim:indi}), because the smallest sum member cannot be expressed in terms of the sum. \\
We see that there are bounds for the regularization, that always either favor an individual or group base-learner only knowing the largest and smallest non-zero singular values of the group matrix and the column vectors as well as the group size. Especially no assumptions regarding the association between predictors, grouped or individual, and the error term were made. Also, the number of boosting iterations performed and the learning rate play no role in the selection bounds. By restricting the design matrix one can find even tighter bounds in which both individual and group selection can happen for a given $\alpha$.
\begin{corollary}\label{corollary:bounds}
    Consider the same setting as in Theorem \ref{th:consistency}. Setting $X = UD$ meaning $V = I_p$, yields the following bound
    \begin{align*}
        \big(\forall_{j \leq p}: \text{df}(\mu) \leq \text{df}(\lambda_j)\big) \Rightarrow \min_{j \leq p} RSS(\lambda_j) \leq RSS(\mu),
    \end{align*}
    and in the case of $X=UdV^T$ with $d \in \mathbb{R}^+$
    \begin{align*}
        \big(\forall_{j \leq p}: \text{df}(\mu) \geq \frac{1}{p}\text{df}(\lambda_i)\big) \Rightarrow \min_{j \leq p} RSS(\lambda_j) \geq RSS(\mu).
    \end{align*}
\end{corollary}
Follows directly from the proof of Theorem \ref{th:consistency}. These bounds have strong implication for setting the mixing parameter $\alpha$, because set too high or too low one only gets either individual or group base-learners for every design matrix. One example of using Corollary \ref{corollary:bounds} is a categorical variable that is treated as a group base-learners and each dummy variable as an individual base-learner. In this case, $\alpha>0.5$ should be set. This is also an example where the condition of \cite{hofner_framework_2011} of setting the degrees of freedom equal across all baselearners fails, as in this case only individual base-learners can be selected, which is a clear case of variable selection bias. If the design matrix of one group is a scaled orthogonal matrix then $\alpha \in [\frac{1}{1+p}, \frac{1}{2}]$ should be set. One example of this would be a categorical base-learner with equal number of observations per category. The following Theorem allows us to characterize the pairwise selection probability in one boosting step of two base-learners where one is a sub-matrix of the other for scaled orthogonal matrices.
\begin{theorem}\label{theo:ortho_probability}
Let $X \in \mathbb{R}^{n \times p}$ be a scaled orthogonal design matrix such that $X = dU$ for $d\in \mathbb{R}^+$ and $U^{n \times p}$ orthogonal. Define the sub-matrix $X^{(1)} \in \mathbb{R}^{n \times p_1}, 0 < p_1 < p $.
Let $y \in \mathbb{R}^n$ be the outcome variable, $y = \epsilon, \epsilon \sim \mathcal{N}(0,\sigma^2)$ not being dependent on the design matrix. Let $\widehat{\beta}_\lambda$ be the Ridge estimate using the design matrix $X^{(1)}$ for some penalty $\lambda>0$ and $\widehat{\beta}_\mu$ the Ridge using $X$ as design matrix for penalty $\mu > 0$. Let $\text{df}(\lambda)$ and $\text{df}(\mu)$ be the corresponding degrees of freedom. If $\frac{\text{df}(\lambda)}{p_1} \geq \frac{\text{df}(\mu)}{p}$ we can characterize the selection probability based on the residual sum of squares for the two base-learners as
\begin{equation*}
    P\big(RSS(\widehat{\beta}_\lambda)\geq RSS(\widehat{\beta}_\mu)\big) = F_{\beta^{'}\big(\frac{p_1}{2}, \frac{p-p_1}{2}, 1, \frac{\text{df}(\lambda)p}{\text{df}(\mu)p_1}-1\big)}(1),
\end{equation*}
where $F_\beta$ is the distribution function of the beta prime distribution.
\end{theorem}
Theorem \ref{theo:ortho_probability} allows us to know the selection probability of a group base-learner versus one individual base-learner for an orthogonal group design matrix. While this is interesting, one would assume that groups are rather defined such that there are dependency structures within a group. However, it is plausible to assume that an individual base-learner from another group is orthogonal to the group design matrix. In that case it is straightforward to see that the selection probability of the individual base-learner vs. the group base-learner follows a generalized beta prime distribution $\beta^{'}\Big(\frac{1}{2},\frac{p-1}{2},1,\frac{\text{df}(\lambda)p}{\text{df}(\mu)}\Big)$, which in the sparse group boosting becomes $\beta^{'}\Big(\frac{1}{2},\frac{p-1}{2},1,\frac{\alpha p}{1-\alpha}\Big)$. 
\subsection{Within-group and between-group selection}
When defining group and individual base-learners at the same time there are two types of variable selection happening at the same time. There is selection between groups: Which group base-learner will be selected? This selection can only be unbiased if there is an equal selection chance for all group base-learners. The sparse group boosting assures them all to have the same degree of freedom, hence the expected RSS is the same for all base-learners. The same is the case for all individual base-learners compared to each other as they also have the same degree of freedom. However, in the presence of individual variables, there will always be a challenge at the group level. To illustrate, consider a categorical variable, one containing 3 categories and a linear base-learner, together building an orthogonal system, not being associated with the outcome variable. Then, if one wants the selection chance of any categorical variable vs the linear variable to be equal their degrees of freedom have to be equal, meaning $\alpha = 0.5$ in the sparse group boosting. But doing this leads to the group base-learner of the categorical variable to be never selected based on Corollary \ref{corollary:bounds}. Furthermore, the individual categorical base-learners will have a greater selection chance compared to the linear base-learner, because of the greater group size. To counter this, one could penalize the individual base-learners in bigger groups more than the ones in smaller groups. Doing this could give the categorical variable vs. linear variable equal selection chances. But then on the individual variable basis, there would be a bias towards selecting an individual base-learner just because the group it belongs to has a smaller group size. Using the sparse group boosting, one can decide if one rather want a balance between groups or between individual base-learners. However one should keep in mind, that a perfect balance in the case of unequal group sizes may not be easy to achieve. On which level one wants equal selection chances depends on the research question and interpretation of the data. Generally, varying group sizes impose challenges. If a group contains only one variable then the group and individual base-learner are the same and therefore the greater value of either $\alpha$ or $1-\alpha$ is used for both, preferring either the group base-learner or the individual base-learner. From Theorem \ref{th:consistency} we see that the group size affects the selection bounds. This can also be seen in Figure \ref{figure:withingroup} which compares the selection frequency of the group base-learner in the first boosting iteration for group sizes two and three and different dependency structures. Staying again with the example of two categorical variables of equal number of observations within each category from Corollary \ref{corollary:bounds} we know that the selection interval of the smaller categorical variable is a subset of the selection interval of the bigger categorical variable. This means that for a small enough $\alpha$ one can either have the group or individual variable selected depending on the dataset and for the smaller group only group base-learners regardless of the dataset. One could use a group adjustment by the group size to align the lower bound but then the upper bound would also be affected imposing the same issue. 
\begin{figure}[!ht]
        \centering
        \includegraphics[scale=0.55]{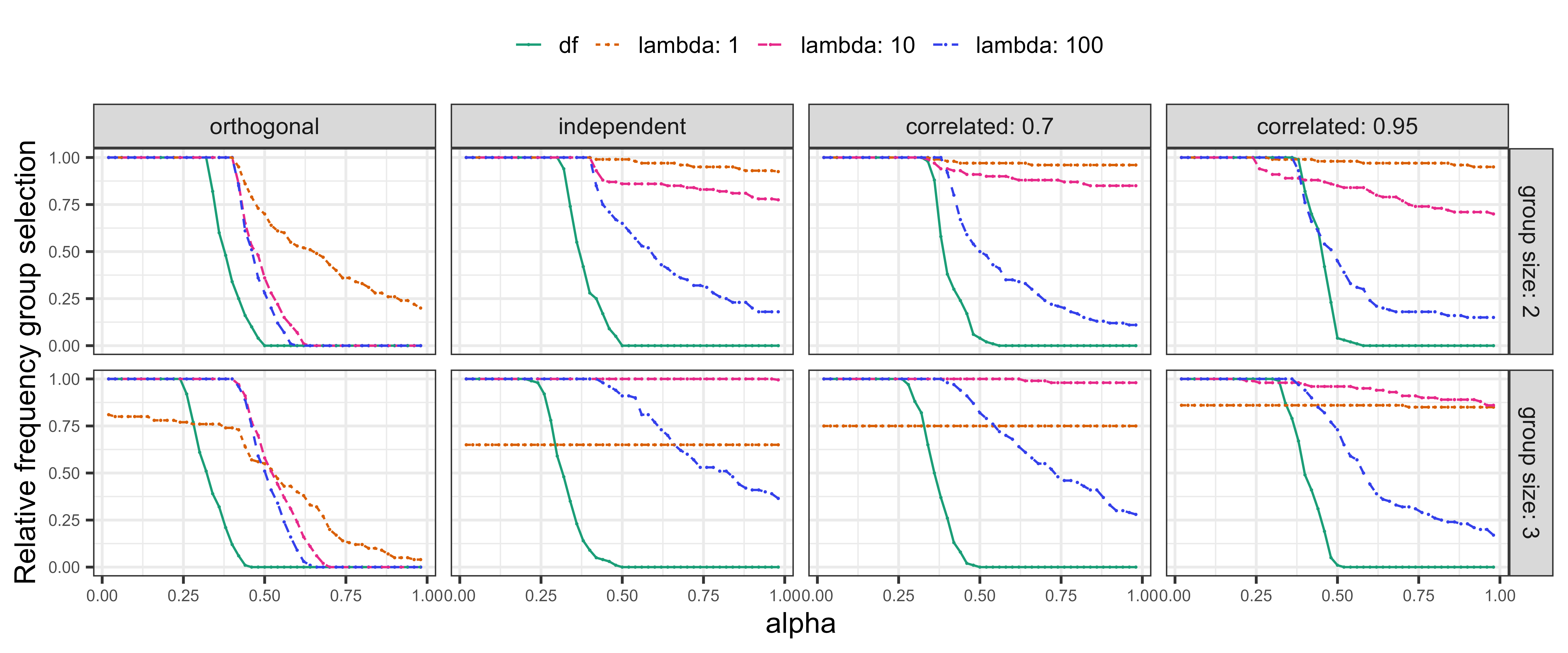}
        \caption{Group selection probability vs individual variables depending on the mixing parameter (alpha) in the sparse group boosting. The two variables within the group are either orthogonal, independent of each other or correlated. }\label{figure:withingroup}
\end{figure}
\FloatBarrier
%
%
\subsection{Extensions}
The here presented results like the selection bounds were mainly focused on $\mathcal{L}^2$ boosting. However, the definition of the sparse group boosting can be applied to many cases of grouped datasets. Whenever the degrees of freedom can be computed and modified by a regularization parameter, one can use sgb df as defined in (\ref{def:sgb_alpha}), and whenever Ridge regularized regression can be used sgb lambda as defined in (\ref{def:sgb_lambda}) can be used. This includes generalized linear and additive models but also regularized regression trees. We also want to highlight that semi-grouped datasets can be analyzed using the sparse group boosting by including additional base-learners which are not split up into groups and individual variables. Examples of this could be random effects, treatment effects, or smoothing splines. The degrees of freedom or regularization parameters of these variables could then either be set to $\alpha$, $1-\alpha$ or even zero yielding an unregularized base-learner.
An extension to generalized additive models for location scale and shape (gamlss) \cite{stasinopoulos_generalized_2008} and their boosting variant 'gamboostLSS in\cite{hofner_gamboostlss_2014} is also possible. This would allow the data analyst to also apply sparse group penalization to the linear predictor for other moments of the conditional distribution of the outcome given covariates. We believe that (group) - sparsity is especially important for higher-order moments due to the overall model complexity and the number of variables to be interpreted. 

%
%

\section{Empirical Data: Agricultural dataset}
The analysis was performed with  R (\cite{r_core_team_r_nodate}). For visualizations, the R package ggplot was used (\cite{wickham_ggplot2_2016}).
All computations were conducted on a 3600 MHz Windows machine.
Biomedical data are prominent examples of where sparse-group selection can be used. To show the variety of possible applications we analyse an agricultural dataset. Climate change impacts on the agricultural sector are well documented. The type and level of impacts are crop and region-specific. Not surprisingly, exposure to climate change makes many orchard farming communities in Chile and Tunisia vulnerable to climate change impacts. Many susceptibility-related factors may affect farm vulnerability to climatic impacts. Several adaptation resources (measures/tools) are available to directly reduce the impact on farm operations or reduce the number or sensitivity of susceptibility-related factors. The final objective is to increase the resilience of the farming communities. \\
The dataset (\cite{pechan_reducing_2023}) contains 12 binary outcome variables of interest that are related to adaptive measures against climate change impacts. The 147 independent variables can be grouped into 23 groups depending on the construct the variable belongs to. Two group examples are social variables as well as past adaptive measures.
801 farmers have been included in the study. Further analysis of group variable selection for other outcomes in the dataset can be found in \cite{obster_financial_2024} and \cite{obster_using_2023}.

We again use 11 equally spaced $\alpha$ values from zero to one as mixing parameters for sgl, sgb df and sgb lamda. The dataset was split into two, each only containing farmers of one Country. We randomly split 70 percent of the data into the training data and 30 percent into the test data. The remaining test data was used for the model evaluation. As in the previous section we used the area under the curve (auc) as an evaluation criterion, since all outcome variables are binary. For the training data, we used a 3-fold cross-validation to estimate the optimal stopping parameter for the boosting models and the optimal $\lambda$ value for the sgl models. We used 6 values of $\lambda$ for the sgl and 3000 boosting iterations with a learning rate $\nu = 0.05$ for both sparse group boosting models. In sgb lambda we used $\lambda = 100$.  

Referring to Figure \ref{figure:envcomp} which averages across all 12 outcome variables for each dataset and $\alpha$ value, it becomes apparent, that overall the models performed similarly regarding predictability. In Chile and Tunisia, sgb df had the highest AUC, obtained at $\alpha = 0.4$ in Chile and $\alpha = 0.2$ in Tunisia. At the same $ \alpha$ values also the sgl achieved its highest auc. Stronger differences between the models can be found regarding sparsity. For smaller $\alpha$ values all models selected more variables on average, where sgb df and sgb lambda selected more variables for smaller $\alpha$ values than sgl. Component-wise boosting and the lasso yielded roughly the same number of selected variables. The number of partially selected groups, meaning at least one variable within a group gets selected, does on average increase with $\alpha$. Whereas the number of fully selected groups decreases on average with increasing $\alpha$ values for sgb df and sgl. This effect is less pronounced for sgb lambda. This opens up an interesting discussion on what "between group sparsity" means. If one defines it through the number of fully selected groups, meaning all variables within the group have to be selected, then compared to defining it through the number of partially selected groups one gets an opposing effect of $\alpha$. The average percentage of selected variables within groups decreases on average with $\alpha$, corresponding to increasing "within-group-sparsity". \\
The computation time was somewhat volatile, and we had to rerun the models a few times, as the sparse group lasso cross-validation estimation with the sgl package crashed multiple times. We did not fully optimize for computational speed and there are efforts to improve both the computation time of sgl (\cite{ida_fast_2019}; \cite{zhang_efficient_2020}) and boosting (\cite{staerk_randomized_2021}). Theoretically, the computation time of the sparse group boosting should be the sum of the time it takes to fit component-wise boosting and group-component-wise boosting. However, through fitting group base-learners parallel to individual learners, the computation time should be close to either group boosting or component-wise boosting depending on which of the two is slower. Therefore, the speed of the sparse group boosting depends mostly on the implementation of boosting.  
\begin{figure}[!ht]
        \centering
        \includegraphics[scale=0.55]{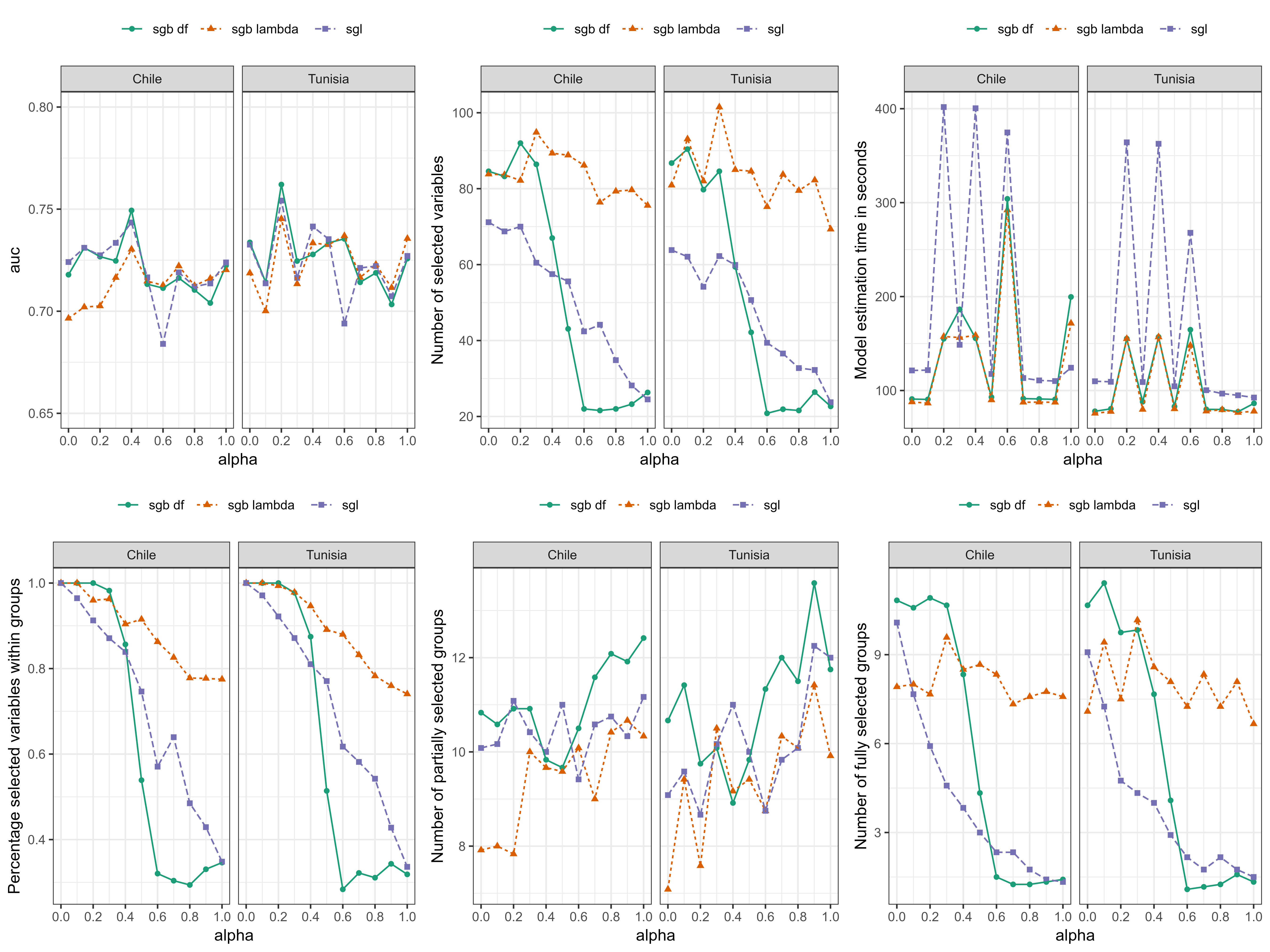}
        \caption{Results of the agricultural dataset. Line-type and point-shape indicate the type model depending on alpha.}\label{figure:envcomp}
\end{figure}
\FloatBarrier
%
%
\section{Simulated data}
In this section, we will compare the two versions of the sparse group boosting with the lasso, to see how similar the predictive power and sparsity properties are. Therefore, we consider 11 equally spaced mixing parameters $\alpha$ between zero and one for the sparse group lasso and sparse group boosting. This way the lasso/boosting and group lasso/group boosting are also covered. \\
The covariate matrix X was simulated with different numbers of covariates, groups, observations, and covariance structures. The response, y was set to
\begin{equation*}
    \sum_{g=1}^G X^{(g)} \beta^{(g)} + \sigma \epsilon.
\end{equation*}
Here, $\epsilon \sim \mathcal{N}(0, I)$. The signal-to-noise ratio between the non-zero entries of $\beta$ and $\sigma \epsilon$ was set to 4 through the value of $\sigma$. Note that the effective signal-to-noise ratio is additionally altered by setting some elements of $\beta$ to zero, which additionally increases the noise. In the case of no variables being associated with the outcome, no additional error term $\epsilon$ was used.
\\
The tuning of the models was performed with a 3-fold cross-validation performed on the whole simulated data. We used 11 equally spaced mixing parameters $\alpha$ ranging from zero to one. For the sparse group boosting based on $\lambda$ and the sparse group lasso, we chose 10 values for $\lambda$. Since no proven method of selecting a good set of $\lambda$ values in the sparse group boosting exists yet, we chose $\lambda = 50 \cdot i$ for $i \in \{1,...,10\}$, as in boosting ridge regression in general bigger values for $\lambda$ are generally preferable \cite{tutz_boosting_2007}. For the boosting models, we used a learning rate of 0.05 and 2500 boosting iterations to fit the models with early stopping derived from a 3-fold cross-validation. Since the sparse group boosting using the degrees of freedom has no comparable tuning parameter for $\lambda$ in the other two models, we used a finer grid of $\alpha$ values. For a given alpha value $\alpha$ in the sparse group lasso, whenever the model is fitted for $\lambda_i$, the sparse group boosting with the degrees of freedom is fitted with $\alpha+0.01\cdot (i-1)$. This way, for each $\alpha$, 10 versions of each of the three models are being fitted. Afterward, we compare the estimates with the actual parameter vector $\beta$. The parameters used for each scenario are summarized in Table \ref{table:sim_params}. For each scenario, 15 iterations of the data were simulated. 
\begin{table}[ht!]
\caption{Table with aligned units. Full gr. refers to the number of groups where each variable inside it has an effect. Half gr. refers to the number of groups that contain exactly five effects and the remaining ones are zero. Empty gr. refers to the number of groups that contain no effects. The number of variables within these groups is described by full vars, half vars, and empty vars. Cor refers to the degree of multicollinearity of the design matrix, measured by the pairwise correlation between the variables in the design matrix.}
  \begin{center}
    \label{table:sim_params}
     \begin{tabular}{c c c c c c c c c}
      Scenario & full gr. & half gr. & empty gr. & full vars & half vars  & empty vars & cor & n \\
       \hline
        1&5&5&5&15&15&15&0&50  \\
        2&5&5&5&5&5&15&0&50  \\
        3&5&5&5&5&15&5&0&50  \\
        4&5&5&5&15&5&5&0&50  \\
        5&2&2&5&15&15&15&0&50  \\
        6&5&2&2&15&15&15&0&50\\
        7&2&5&2&15&15&15&0&50 \\
        8&0&0&5&0&0&15&0&50     \\
        9&5&0&0&15&0&0&0&50\\
        10&5&5&5&15&15&15&0&500 \\
        11&5&5&5&15&15&15&0.5&50\\
        12&5&5&5&15&15&15&0.95&50 \\
\end{tabular}
  \end{center}
\end{table}
As the main evaluation criterion, we used the root mean squared error (RMSE), defined as
\begin{equation*}
    \text{RMSE} = \frac{1}{p}\sqrt{\sum_{g=1}^G \sum_{j = 1}^{p_g}\big{(}(\beta^{(g)}_j-\widehat{\beta}_j^{(g)}\big{)}^2},
\end{equation*}
where $X_{i \cdot}$ is the $i$-th row of the design matrix. We also computed the proportion of "correct effects" \begin{equation*}
    \frac{\sum_{g=1}^G \sum_{j = 1}^{p_g}\mathds{1}_{[\beta^{(g)}_j \neq 0 \land \widehat{\beta}_j^{(g)} \neq 0]}}{\sum_{g=1}^G \sum_{j = 1}^{p_g}\mathds{1}_{[\beta^{(g)}_j \neq 0 ]}},
\end{equation*}
and the proportion of "correct zeros" and the overall "correct classified" elements of $\beta$.
In Figure \ref{figure:sim_res}, the results of the simulation are displayed. For each model type and $\alpha$ value, out of the 10 hyperparameters, the model with the lowest RMSE was chosen. For each metric considered the values were then averaged across the 15 iterations. Since sgb df had a finer grid of $\alpha$ values, we removed the second digit of $\alpha$, e.g. 0.47 becomes 0.4. The only exceptions were the values between 0.01 and 0.10, which we rounded up to 0.1. This is because $\alpha = 0$ is group boosting and we did not want to mix it with the sparse group boosting. The results of sgb df on the full scale of $\alpha$ without summarizing are shown in Figure \ref{figure:sim_alpha}. \\
Generally, in most scenarios and models, sparse group variable selection improves the fit. For $\alpha \in \{0,1\}$ (group lasso, group boosting, lasso and component-wise boosting), sgb df and sgl yield similar estimates, except scenarios 8 and 10, as the RMSE and also the detection rates are close together. This is in line with \cite{hepp_approaches_2016}. However, the effect of $\alpha$ on the evaluated metrics differs. This is in line with the results from Theorem \ref{th:consistency} and Corollary \ref{corollary:bounds}, as for $\alpha \geq 0.5$, sgb df will be close to component-wise boosting with varying degrees of freedom. This is a difference to sgl, in which the resulting model changes through the whole range of $\alpha$. Generally, there is a trade-off between the correct detection of effects and zeros, which is affected by $\alpha$ through the selection of either groups or single variables. This is the case for all covered models. The range of the correct detection of zeros and effects is greater for sgb df than all other models, which can partly be explained by a finer grid of $\alpha$ values. In this regard, it is important to note again, that sgb lambda is not guaranteed to yield only group-wise selection as seen in Figure \ref{figure:withingroup}. 
In scenario 10 sgl outperformed both variations of the sparse group boosting. However, in this scenario, sgb df was almost always stopped out meaning the number of boosting iterations was set too small for this dataset. A similar issue happens with the models fitted with the 'SGL' package, as the vector of lambda values contained too small values leading to severe over-fitting (compare with Figure 1 in Appendix B). Therefore, we used the 'sparsegl' package for the simulations, which chooses more sensible values for lambda in the cross-validation. \\
  As in many cases, there is no "the best model", only the best model for a given metric and dataset, especially if one looks at opposing metrics.
For sgb df and sgl, whenever there are more full groups or the number of variables in the full groups is greater compared to the half and empty groups (Scenarios 4,6 and 9), $\alpha$ decreases the correct detection of $\beta$ elements and increases RMSE, meaning group-wise selection is more important than individual variable selection. The opposite is true when the number of half groups and the group size of the half groups is greater (Scenarios 3 and 7). Multicollinearity (Scenario 1 vs Scenarios 11 vs. 12) the correct detection rates are similar for the lasso and component-wise boosting. It seems to affect the models for smaller values of $\alpha$ more. As also apparent in Figures \ref{figure:withingroup} and \ref{figure:sim_alpha}, the bounds of the interval in which group selection and individual variable selection happen together in the sgb df is shifted to the right in the case of multicollinearity. \\
\begin{figure}[!ht]
        \centering
        \includegraphics[scale=0.48]{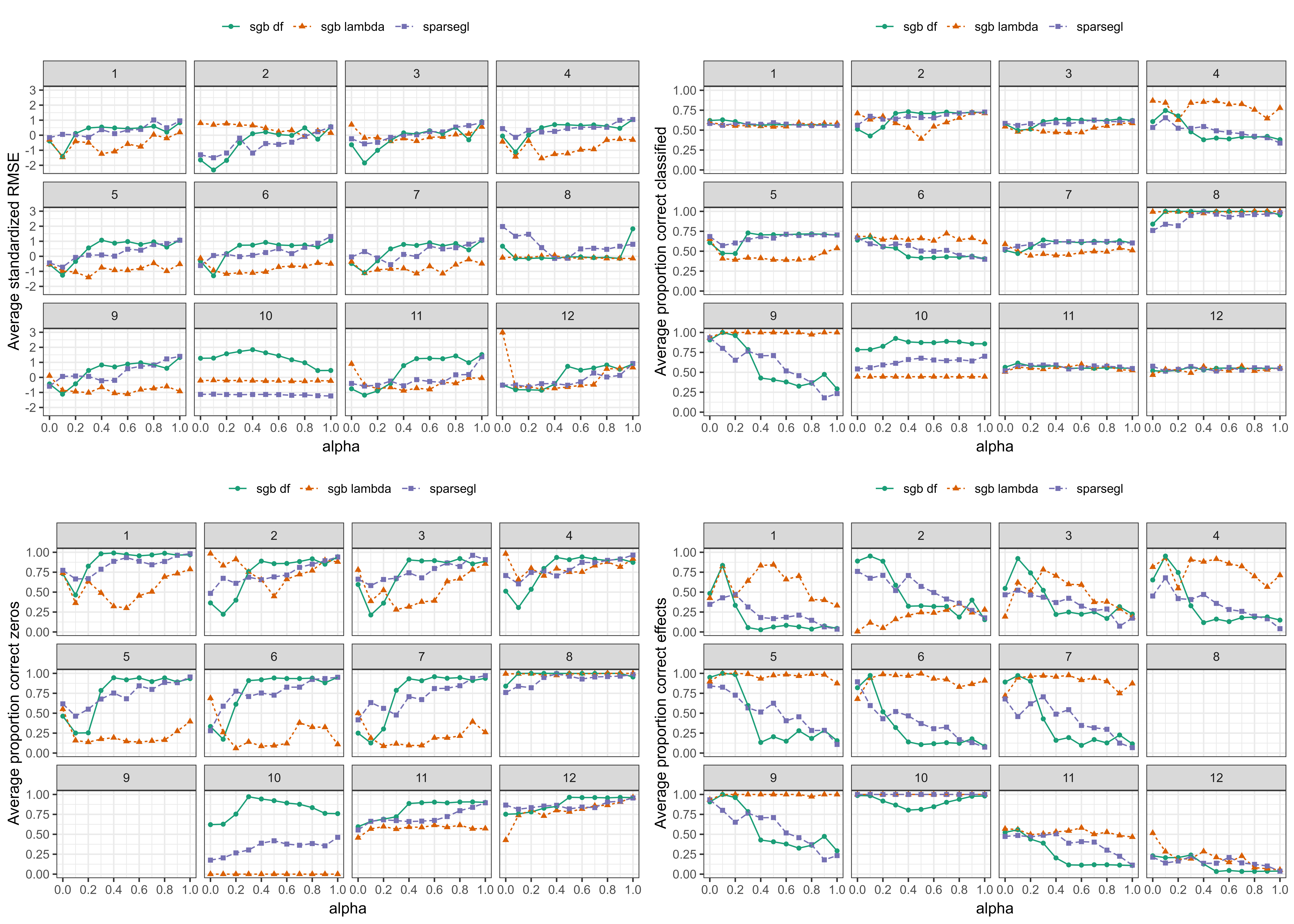}
        \caption{Simulation results for the 12 simulated scenarios averaged across the 15 iterations and 10 hyper-parameter setting for each alpha. Line-type and point-shape indicates the type of model. All metrics compare the model estimates with the true parameter vector.}\label{figure:sim_res}
\end{figure}

\FloatBarrier
\begin{figure}[!ht]
        \centering
        \includegraphics[scale=0.95]{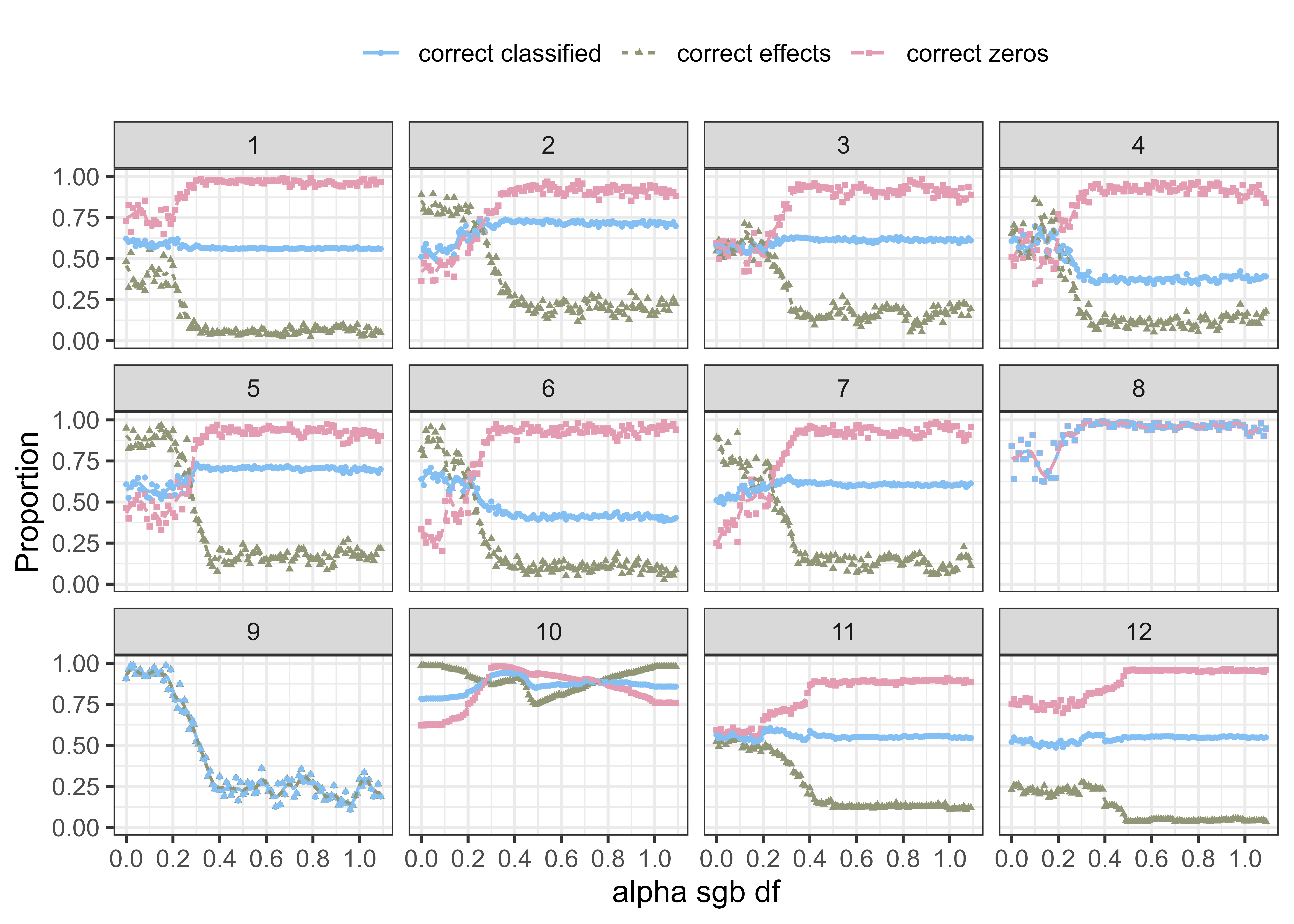}
        \caption{Effect of alpha on the proportion of correct classification of effects and zeros in the sparse group boosting (sgb df) for the 12 simulated scenarios. Each point represents an average of the 15 iterations. Line-type and point-shape indicates the type of detection rate comparing the estimates with the actual parameter vector $\beta$. Smoothed lines were computed by the LOESS with a span of 0.1.}\label{figure:sim_alpha}
\end{figure}
\FloatBarrier
\section{Conclusion, limitations, future work and discussion}
In this manuscript, we presented a framework for adapting sparse group variable selection to boosting. Combining group and individual base-learners in the same model, by weighting the degrees of freedom as mixing parameter $\alpha$, one can fit a model with similar characteristics as the sparse group lasso. However, the effect of $\alpha$ is different in the sparse group lasso compared to $\alpha$ in the sparse group boosting. Even though both models yield only group selection for $\alpha = 0$ and only individual variable selection for $\alpha = 1$, there is a greater range where the sparse group boosting only selects individual variables or groups. We found theoretical bounds for this range depending on the singular values of the group design matrix in $\mathcal{L}^2$ boosting. This implies that for model tuning one should focus on the $\alpha$ values within the theoretical bounds. A good proxy without computing the singular values a priori is to use $\alpha \in [\frac{1}{p_{\text{max}+1}}, 0.6]$. This interval is the one from Corollary \ref{corollary:bounds} and gives some room for multicollinearity and strongly varying singular values of a group design matrix. In the simulations as well as the two real datasets this range was sufficient. While finding the right value for $\alpha$ in the sparse group boosting, it has a natural interpretation in addition to just being a mixing parameter, as it corresponds to the degrees of freedom. \\
Mixing the ridge regularization parameter directly (sgb lambda) one can also fit a sparse group boosting. However, the fitting becomes harder as one has an additional hyper-parameter which has a strong effect on the selection bounds and one looses the interpretability of $\alpha$. Further research would be required to make this formulation a useful competitor. \\
We want to note that there are also other possible ways to fit a similar model through boosting. One idea is boosting the sparse group lasso, or using group boosting with an elastic net penalty within each group-base-learner.
The main difference between the sparse group boosting and the sparse group lasso is the fitting philosophy. When thinking about the sparse group lasso one thinks of shrinking effects and making them vanish either on a group level or on an individual variable level. When thinking about boosting, one rather think about adding individual variables or groups of variables to the global model. Boosting the sparse group lasso or group boosting the elastic net would combine both approaches iteratively while adding and shrinking at the same time. This could have an interesting and different selection behavior over the here proposed sparse group boosting, as one does not have to update full groups if a group base-learner is selected. A distinct advantage of the philosophy of adding rather than shrinking of the sparse group boosting is that one can have a group being selected, updating all variables equally and also additional individual variables on top of the group, which are more important than the other variables within the group. This way one can compare the variable importance of individual variables vs. group variables through the proportion of explained variance/loss function which cannot be done by the sparse group lasso and boosted variants of it. An example of this can be found in \cite{obster_financial_2024}. This way the sparse group boosting can facilitate new research questions and provide additional insight and interpretability. \\
\printbibliography
\appendix
\section{Theorms with proofs}
\begin{theorem}[Distribution of the difference of RSS in orthogonal Ridge regression]\label{theo:ortho_gamma_supp}
Let $X \in \mathbb{R}^{n \times p}$ be a design matrix with orthonormal columns such that $X^TX = I_p$.
Let $y \in \mathbb{R}^n$ be the outcome variable, $y = \epsilon, \epsilon \sim \mathcal{N}(0,\sigma^2)$ not dependent on the design matrix. Further, assume that the least squares estimate $\widehat{\beta} = X^Ty$ exists and $\widehat{\beta}_\lambda$ is the Ridge estimate for some $\lambda>0$.
Define the difference of residual sums of squares as $\Delta =  \text{RSS}(\widehat{\beta}_\lambda) - \text{RSS}(\widehat{\beta}) = (y- \frac{1}{1+\lambda}X \widehat{\beta})^{T}(y-\frac{1}{1+\lambda}X \widehat{\beta}) - (y-X \widehat{\beta})^{T}(y-X \widehat{\beta})$. Then if $\big(1-\frac{\text{df}(\lambda)}{p}\big)> 0$, $\frac{\Delta}{\sigma^2}$ follows a gamma distribution with the following shape-scale parametrization
\begin{equation*}
    \frac{\Delta}{\sigma^2} \sim \Gamma \Big(\frac{p}{2}, 2(1-\frac{2}{p(1+\lambda)}+\frac{1}{p(1+\lambda)^2})\Big).
\end{equation*}
\end{theorem}
\begin{proof}
    \begin{align*}
        \Delta &= y^T\Big(I_p-\frac{1}{1+\lambda}XX^T\Big)^2y-y^T(I_p-XX^T)y \\
        &= y^T\Big(-\frac{2}{1+\lambda}XX^T+\frac{1}{(1+\lambda)^2}XX^T+XX^T\Big)y \\
        &= \Big(1-\frac{2}{1+\lambda}+\frac{1}{(1+\lambda)^2}\Big)y^TXX^Ty \\
        &= \Big(1-\frac{\text{df}(\lambda)}{p}\Big)y^TXX^Ty.
    \end{align*}
    Since $1-\frac{\text{df}(\lambda)}{p}> 0$ and $\frac{y^TXX^Ty}{\sigma^2} \sim \chi^2(p)$ we end up with $\frac{\Delta}{\sigma^2} \sim \Gamma \Big(\frac{p}{2}, 2(1-\frac{\text{df}(\lambda)}{p})\Big)$.
\end{proof}
\begin{theorem}[Selection intervals of the sparse group boosting]
    Consider a design matrix $X \in \mathbb{R}^{n \times p}$ of rank $r \leq p$ with singular value decomposition $X = UDV^{T}$, where $U \in \mathbb{R}^{n\times p}, V \in \mathbb{R}^{p\times p}$ are unitary matrices and $D = \text{diag}(d_1,...,d_r,0,...,0)$ is a diagonal Matrix containing the singular values. Let $y \in \mathbb{R}^n$ be the outcome variable and $\widehat{\beta}_{\mu} = (X^{T}X+\mu I)^{-1}X^{T}y $ be the Ridge estimate for $\mu \geq 0$. For $j \leq p$ let $\widehat{\beta}_{\lambda_j} = (x_j^{T}x_j+\lambda_j)^{-1}x_j^{T}y$ be the estimate for the $j$-th individual base-learner, and $\overline{d}_j$ be the singular value of $x_j$. Denote $\overline{d}^- = \min_{j \leq p} \overline{d}_j^2$ and $\overline{d}^+ = \max_{j \leq p} \overline{d}_j^2$ as well as and $d^+ = \max_{j \leq r} d_j^2$ and $d^- = \min_{j \leq r} d_j^2$ accordingly. Then, there are always two mixing parameters $\alpha_1, \alpha_2 \in ]0,1[$ such that
    \begin{align*}
        (\forall_{j \leq p}:  \alpha_1 &= \text{df}(\lambda_j) \wedge (1-\alpha_1) = \text{df}(\mu)) \Rightarrow \min_{j \leq p}{RSS(\lambda_j)} \leq RSS(\mu), \text{and} \\
        (\forall_{j \leq p}: \alpha_2 &= \text{df}(\lambda_j) \wedge (1-\alpha_2) = \text{df}(\mu)) \Rightarrow  RSS(\mu) \leq \min_{j \leq p}{RSS(\lambda_j)}.
    \end{align*}
    Furthermore, the following conditions assure the selection of an individual variable or the whole design matrix
    \begin{align}
        \bigg( \bigg[\forall_{l \leq k} \frac{(d^-+2\mu)}{(d^-+\mu)^2} \leq \frac{(\overline{d}_l^2+2\lambda_l)}{r(\overline{d}_l^2+\lambda_l)^2}\bigg] \lor \bigg[\text{df}(\mu) \leq \frac{\text{df}(\lambda_l)d^-}{r\overline{d}^+}\bigg] \bigg) &\Rightarrow \min_{j \leq p}{RSS(\lambda_j)} \leq RSS(\mu), \label{claim:indi_supp} \\
         \bigg(\bigg[\forall_{l \leq k} \frac{(d^++2 \mu)}{(d^++\mu)^2} \geq \frac{\text{df}(\lambda_l)}{\overline{d}_l^2}\bigg] \lor \bigg[\forall_{l \leq k} \frac{(d^++2 \mu)}{(d^++\mu)^2} \geq \frac{(\overline{d}_l^2+2 \lambda_l)}{(\overline{d}_l^2+\lambda_l)^2}\bigg]) &\Rightarrow  RSS(\mu) \leq \min_{j \leq p}{RSS(\lambda_j)}.\label{claim:group_supp}
    \end{align}
\end{theorem}
\begin{proof}
    First, choose a base-learner with a design matrix vector denoted as $x_l$. By using the singular value decomposition of $x_l=\overline{u}_l\overline{d}_l$ and $X=UDV^T$, we can rewrite
    $RSS(\lambda_l)$ and $RSS(\mu)$ as
    \begin{align*}
       RSS(\mu) = y^Ty -y^T\Big(\sum_{j=1}^r \Big[2\frac{d_j^2}{d_j^2+\mu}-\frac{d_j^4}{(d_j^2+\mu)^2}\Big]u_ju_j^{T}\Big)  y
    \end{align*}
    and
    \begin{align*}
       RSS(\lambda_l) = y^Ty -y^T\Big(\Big[2\frac{\overline{d}^2_l}{\overline{d}^2_l+\lambda_l}-\frac{\overline{d}_l^4}{(\overline{d}^2_l+\lambda_l)^2}\Big]\overline{u}_l \overline{u}_l^{T}\Big) y = y^Ty- y^T \text{df}(\lambda_l) \overline{u}_l\overline{u}_l^Ty.
    \end{align*}
    Denote the diagonal elements $\tilde{d}_j = 2\frac{d_j^2}{d_j^2+\mu}-\frac{d_j^4}{(d_j^2+\mu)^2}$ of $\tilde{D}$ and
    note that $\overline{d}_j^2 \in [d^-,d^+]$, because $x_l$ is a sub-matrix of $X$.\\ 
    Then for some $l \leq p$
    \begin{align}
        RSS(\mu) - RSS(\lambda_l) &= y^Ty -y^T\Big(\sum_{j=1}^r \tilde{d}_j u_ju_j^{T}\Big)  y - \big(y^Ty -y^T \text{df}(\lambda_l) \overline{u}_l\overline{u}_l^{T}y \big) \nonumber\\
        &= y^T\Big[-\sum_{j=1}^r \tilde{d}_j u_ju_j^{T} + \text{df}(\lambda_l) \overline{u}_l\overline{u}_l^{T} \Big]y \nonumber\\
        &=y^T\big(-U\tilde{D}U^T+\frac{\text{df}(\lambda_l)}{\overline{d}_l^2} UD(V^T)_l{(V^T)_l}^TDU^T\big)y \nonumber\\
        &=y^TU\big(-\tilde{D}+\frac{\text{df}(\lambda_l)}{\overline{d}_l^2} D(V^T)_l{(V^T)_l}^TD\big)U^Ty.\label{equ:start}
    \end{align}
    Here we have used the fact that $\overline{u}_l$ is the left singular value of $x_l= UD(V^T)_l$.\\ We will now consider the first claim (\ref{claim:indi_supp}). \\
    Using the norm inequality $\forall z \in \mathbb{R}^r:\norm{z}_\infty^2 \geq \frac{1}{r} \norm{z}_2^2$ we get
    \begin{align}
        RSS(\mu) - \min_{l \leq k} RSS(\lambda_l) &=\max_{l \leq k} y^TU\big(-\tilde{D}+\frac{\text{df}(\lambda_l)}{\overline{d}_l^2} D(V^T)_l{(V^T)_l}^TD\big)U^Ty \nonumber \\ 
        &\geq y^TU(-\tilde{D}+\frac{\text{df}(\lambda_{l})}{r{\overline{d}_l}} D^2)U^Ty.\label{condition:firstclaim}
    \end{align}
    Here we have also used that $\frac{\text{df}(\lambda_l)}{{\overline{d}_l}^+}$ is minimal for $l^*$: $\overline{d}_{l^*} =\overline{d}^+$, because $d \mapsto \frac{1}{d^2} \Big(\frac{2d^2}{d^2+\lambda}- \frac{(d^2)^2}{(d^2+\lambda)^2}\Big)$ is a decreasing function which can be seen by taking the derivative with respect to $d$. In the case of $\text{df}(\mu) \leq \frac{\text{df}(\lambda)d^-}{r\overline{d}^+}$ all diagonal elements in (\ref{condition:firstclaim}) are greater zero because $\text{df}(\mu) \geq \tilde{d}_j$. To see the other part of (\ref{claim:indi_supp}) we continue with (\ref{condition:firstclaim}), look at diagonal element $j$ and observe that
    \begin{align*}
        \bigg[\frac{2d_j^2}{d_j^2+\mu}- \frac{d_j^4}{(d_j^2+\mu)^2}\bigg] \leq \frac{\frac{2\overline{d}^2_l}{\overline{d}^2_l+\lambda}- \frac{\overline{d}_l^4}{(\overline{d}^2_l+\lambda_l)^2}}{r\overline{d}^2_l} d_j^2 \\
        \Leftrightarrow \frac{(d_j^2+2\mu)}{(d_j^2+\mu)^2} \leq \frac{(\overline{d}^2_l+2\lambda_l)}{r(\overline{d}^2_l+\lambda_l)^2}.
    \end{align*}
    Therefore, all diagonal elements are greater zero if (\ref{claim:indi_supp}) holds and $d^2_j= d^-$, since $d \mapsto \frac{-(d^2+2 \mu)}{(d^2+\mu)^2}$ is a decreasing function, which can be easily seen by taking the derivative with respect to $d$. \\
    For the second claim in (\ref{claim:group_supp}), we return to(\ref{equ:start}), which can be bounded by
    \begin{align*}
        &RSS(\mu) - RSS(\lambda_l) \\
        = &y^TU(-\tilde{D}+\frac{\text{df}(\lambda_l)}{\overline{d}_l^2} D(V^T)_l{(V^T)_l}^TD)U^Ty. \\ 
        \leq &y^TU\Bigg(-\text{diag}\bigg[\frac{2d^2_j}{d_j^2+\mu}- \frac{(d_j^2)^2}{(d^2_j+\mu)^2}\bigg]+\bigg[\frac{\text{df}(\lambda_l)}{\overline{d}_l^2}\bigg]\text{diag}(d_j^2)\Bigg)U^Ty. \\
        = &y^TU\Bigg(-\text{diag}\bigg[\frac{d^2_j(d_j^2+2 \mu)}{(d_j^2+\mu)^2}\bigg]+\bigg[\frac{(\overline{d}_l^2+2 \lambda_l)}{(\overline{d}_l^2+\lambda_l)^2}\bigg]\text{diag}(d_j^2)\Bigg)U^Ty \\
        \leq  &y^TU\Bigg(\text{diag}(d^2_j)\bigg[\frac{-(d^++2 \mu)}{(d^++\mu)^2}+\frac{(\overline{d}_l^2+2 \lambda_l)}{(\overline{d}_l^2+\lambda_l)^2}\bigg]\Bigg)U^Ty \leq 0.
    \end{align*} 
    In the last step we have used that $d \mapsto \frac{(d^2+2 \mu)}{(d^2+\mu)^2}$ is a decreasing function. \\
    The first part follows from (\ref{claim:indi_supp}) and (\ref{claim:group_supp})
\end{proof}
\begin{theorem}\label{theo:ortho_probability_supp}
Let $X \in \mathbb{R}^{n \times p}$ be a scaled orthogonal design matrix such that $X = dU$ for $d\in \mathbb{R}^+$ and $U^{n \times p}$ orthogonal. Define the sub-matrix $X^{(1)} \in \mathbb{R}^{n \times p_1}, 0 < p_1 < p $.
Let $y \in \mathbb{R}^n$ be the outcome variable, $y = \epsilon, \epsilon \sim \mathcal{N}(0,\sigma^2)$ not being dependent on the design matrix. Let $\widehat{\beta}_\lambda$ be the Ridge estimate using the design matrix $X^{(1)}$ for some penalty $\lambda>0$ and $\widehat{\beta}_\mu$ the Ridge using $X$ as design matrix for penalty $\mu > 0$. Let $\text{df}(\lambda)$ and $\text{df}(\mu)$ be the corresponding degrees of freedom. If $\frac{\text{df}(\lambda)}{p_1} \geq \frac{\text{df}(\mu)}{p}$ we can characterize the selection probability based on the residual sum of squares for the two base-learners as
\begin{equation*}
    P\big(RSS(\widehat{\beta}_\lambda)\geq RSS(\widehat{\beta}_\mu)\big) = F_{\beta^{'}\big(\frac{p_1}{2}, \frac{p-p_1}{2}, 1, \frac{\text{df}(\lambda)p}{\text{df}(\mu)p_1}-1\big)}(1),
\end{equation*}
where $F_\beta$ is the distribution function of the beta prime distribution.
\end{theorem}
\begin{proof}
    \begin{align*}
    \text{RSS}(\widehat{\beta}_\lambda) - \text{RSS}(\widehat{\beta}_\mu) \geq 0 \Leftrightarrow \\
     y^T\Big(I_p-\frac{1}{1+\lambda}X^{(1)}{X^{(1)}}^T\Big)^2y-y^T\Big(1-\frac{1}{1+\mu}XX^T\Big)^2y \geq 0 \Leftrightarrow \\
        -\frac{\text{df}(\lambda)}{p_1}y^TU^{(1)}{U^{(1)}}^Ty + \frac{\text{df}(\mu)}{p}y^TUU^Ty \geq 0 \Leftrightarrow \\
        \text{df}(\mu)y^TU_{\{p_1+1,...,p\}}{U_{\{p_1+1,...,p\}}}^Ty - \Big( \frac{\text{df}(\lambda)}{p_1}-\frac{\text{df}(\mu)}{p}\Big) y^TU^{(1)}{U^{(1)}}^Ty \geq 0 \\
        \frac{\Big( \frac{\text{df}(\lambda)}{p_1}-\frac{\text{df}(\mu)}{p}\Big) y^TU^{(1)}{U^{(1)}}^Ty}{\frac{\text{df}(\mu)}{p}y^TU_{\{p_1+1,...,p\}}{U_{\{p_1+1,...,p\}}}^Ty} \leq 1
    \end{align*} 
In the 8-th line $\frac{\text{df}(\lambda)}{p_1}-\frac{\text{df}(\mu)}{p}$ was used. From the derivations in Lemma \ref{theo:ortho_gamma_supp} we know that 
\begin{equation*}
\frac{\Big( \frac{\text{df}(\lambda)}{p_1}-\frac{\text{df}(\mu)}{p}\Big) y^TU^{(1)}{U^{(1)}}^Ty}{\frac{\text{df}(\mu)}{p}y^TU_{\{p_1+1,...,p\}}{U_{\{p_1+1,...,p\}}}^Ty} \sim \frac{\Gamma\Big(\frac{p_1}{2},2\Big(\frac{\text{df}(\lambda)}{p_1}-\frac{\text{df}(\mu)}{p}\Big)\Big)}{\Gamma\Big(\frac{p-p_1}{2},\frac{2\text{df}(\mu)}{p}\Big)} \sim \beta^{'}\Big(\frac{p_1}{2}, \frac{p-p_1}{2}, 1, \frac{\text{df}(\lambda)p}{\text{df}(\mu)p_1}-1\Big).
\end{equation*}
For the last step, the independence of the two gammas was used which follows from the orthogonality of $X$ and therefore the independence of the two quadratic forms.
\end{proof}
\section{Organisational research data}
As a second application, we use a dataset in the field of innovation research in the public sector conducted in 2020 (https://github.com/FabianObster/sgb). A number of 208 soldiers have been interviewed with a focus on organizational empowerment and its determining factors within the German armed forces. We use 10 groups of variables each containing 4 variables associated with the individual innovation potential (\cite{schiesl_erstellung_2015}) and one group containing 20 variables describing the organizational innovation (Intrapreneurship) potential (\cite{moghaddas_organizational_2020}) to explain the numeric outcome variable "the organizational empowerment scale" 
(\cite{matthews_organizational_2003}). A common way to analyze these types of datasets in the social sciences is to average the variables (items) belonging to a group (construct), as the number of items is relatively high and they are in many cases correlated within a group. However, with this approach, within-group comparisons and sparsity are not obtainable. Models performing sparse-group variable selection allow for more flexibility. We compare the sparse group lasso with the here proposed versions of the sparse group boosting. We are interested in two properties, the predictive performance on held-out data and the sparsity property depending on the mixing parameter $\alpha$. To do this, we fit the three models to half of the data and compare the predictive performance measured by the mean squared error (MSE) on the other half of the data. We also compare the total number of selected variables as a sparsity measure. All variables were standardized. For all models, we used 11 equally spaced mixing parameters $\alpha$ ranging from zero to one. For the sparse group boosting based on $\lambda$ and the sparse group lasso, we chose 10 values for $\lambda$. For the sparse group lasso, we used a 5-fold cross-validation with the function 'cvSGL' from the R package 'SGL' which determines its own values for $\lambda$. Since no proven method of selecting a good set of $\lambda$ values in the sparse group boosting exists yet, we chose $\lambda = 50 \cdot i$ for $i \in \{1,...,10\}$, as in boosting ridge regression in general bigger values for $\lambda$ are generally preferable \cite{tutz_boosting_2007}. For the boosting models, we used a learning rate of 0.01 and 2000 boosting iterations to fit the models with early stopping derived from a 5-fold cross-validation. Since the sparse group boosting using the degrees of freedom has no comparable tuning parameter for $\lambda$ in the other two models, we used a finer grid of $\alpha$ values. For a given alpha value $\alpha$ in the sparse group lasso, whenever the model is fitted for $\lambda_i$, the sparse group boosting with the degrees of freedom is fitted with $\alpha+0.01\cdot (i-1)$. This way, for each $\alpha$, 10 versions of each of the three models are being fitted. We always chose the model with the lowest MSE for each $\alpha$ evaluated on the training data and plotted the results in Figure \ref{figure:empowerment}. We see that for all $\alpha$ values both versions of the sparse group boosting are competitive comparing the MSE and yield a sparser set of selected variables at the same time. However, in this dataset, it looks like the utilization of group structure decreases the predictive power. the lasso outperformed the group lasso and sparse group lasso, and the MSE of both sparse group boosting versions is lower for $\alpha \geq 0.6$ compared to smaller values.

\begin{figure}[!ht]
        \centering
        \includegraphics[scale=0.45]{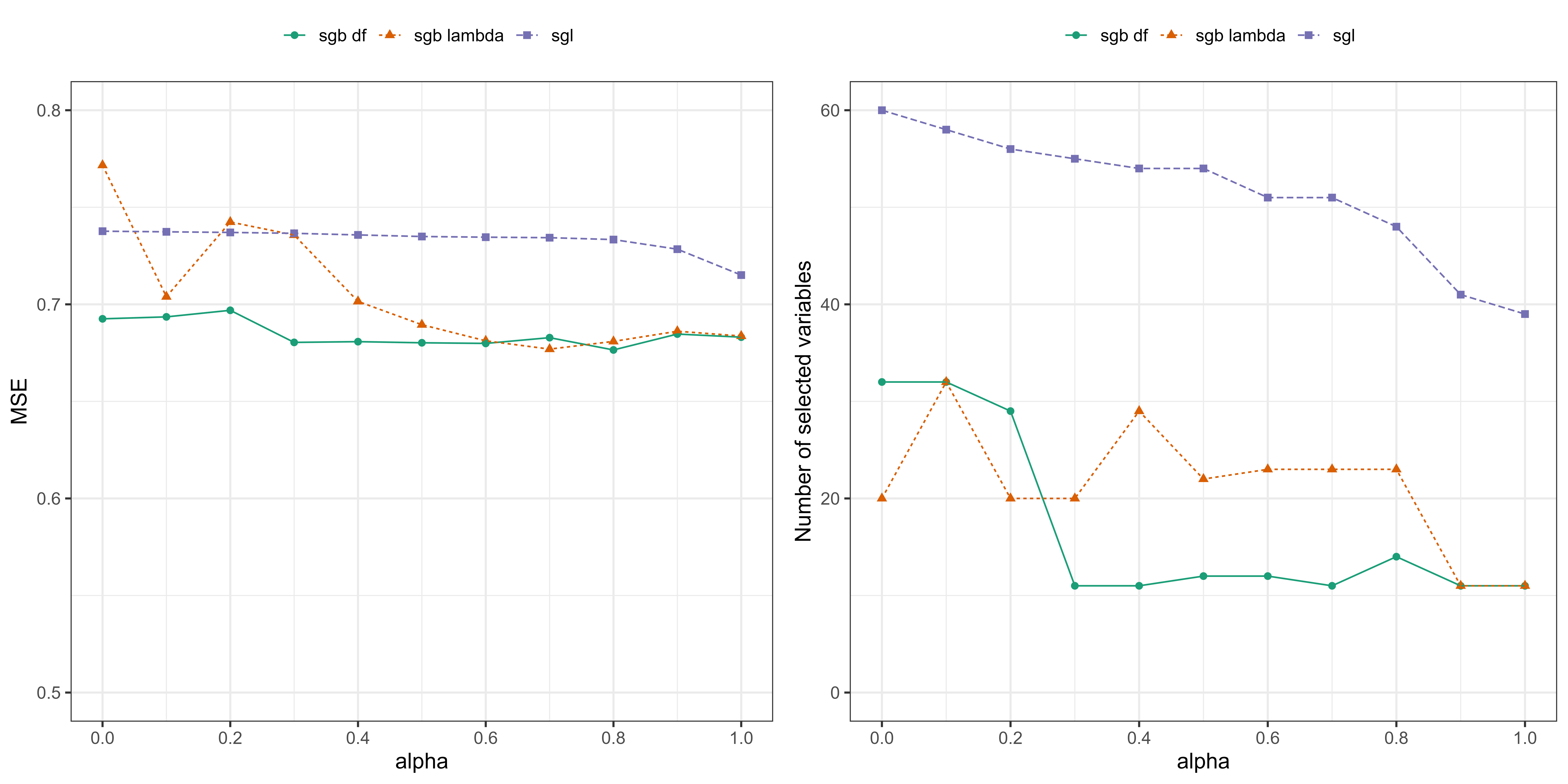}
        \caption{Out of sample MSE (left) and the number of selected variables out of the 60 variables in the dataset (right) for various mixing parameters alpha on the x-axis. Line-type and point-shape indicate the model Sparse group boosting mixing the degrees of freedom (sgb df), sparse group boosting mixing the ridge regularization parameter (sgb lambda), and the sparse group lasso (sgl).}\label{figure:empowerment}
\end{figure}
\FloatBarrier
\section{Further simulation results}

\FloatBarrier

\bigskip
\begin{center}
\begin{figure}[!ht]
        \centering
        \includegraphics[scale=0.48]{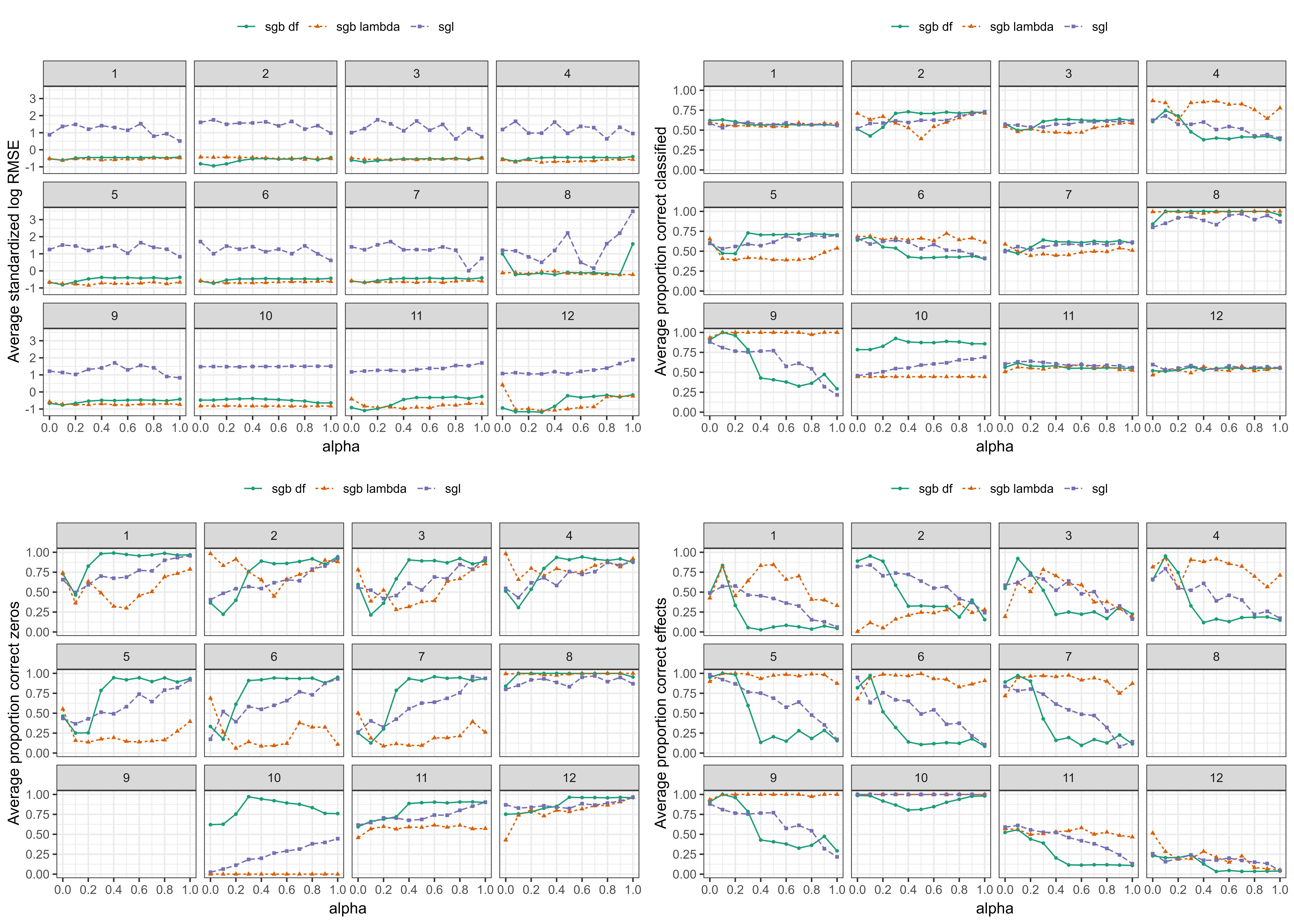}
        \caption{Simulation results for the 12 simulated scenarios averaged across the 15 iterations and 10 hyperparameter setting for each alpha. Colour indicates the type of model. All metrics compare the model estimates with the true parameter vector. Sparse group lasso fitted via the R-package 'SGL'}\label{figure:sim_res_sgl}
\end{figure}

\end{center}

\FloatBarrier

\bigskip
\end{document}